\documentclass[11pt]{article}
\usepackage[utf8]{inputenc}
\usepackage{amsfonts,epsfig}
\usepackage[hyphens]{url}
\RequirePackage{color}
\usepackage{microtype}
\usepackage{mathtools}
\usepackage{enumerate}
\usepackage{tabularx}
\usepackage{pifont} 
\usepackage{xcolor}
\definecolor{darkblue}{rgb}{0,0.08,0.45}
\definecolor{darkred}{RGB}{139,0,0}
\definecolor{Darkblue}{RGB}{0,0,139}
\definecolor{forestgreen}{RGB}{34,139,34}
\definecolor{darkgreen}{RGB}{0,100,0}
\usepackage{tikz}
\usepackage{bbm}
\usepackage{algpseudocode}
\usepackage{algorithm}
\usepackage{stackengine}
\usetikzlibrary{shapes,arrows,chains}
\usetikzlibrary[calc]
\usetikzlibrary{bayesnet}
\tikzset{>=latex}
\usepackage{hyperref}
\hypersetup{
	colorlinks=true,
	linkcolor=black,
	citecolor=black,
	filecolor=black,
	urlcolor=black,
}

\usepackage{breakurl}
\usepackage{comment}
\graphicspath{{fig/}}
\usepackage{amsthm}
\newtheorem{theorem}{Theorem}

\usepackage{fancyvrb}
\let\oldv\verbatim
\let\oldendv\endverbatim
\def\verbatim{\par\setbox0\vbox\bgroup\oldv}
\def\endverbatim{\oldendv\egroup\fboxsep0pt \noindent\colorbox[gray]{0.96}{\usebox0}\par}

\usepackage{amssymb}
\usepackage{geometry}
\usepackage{sectsty}
\usepackage{caption}
\usepackage{subcaption}
\usepackage{amsmath}
\usepackage{tikz}
\usetikzlibrary{positioning}
\DeclareMathOperator{\E}{\mathrm{E}}

\newcolumntype{b}{X}
\newcolumntype{s}{>{\hsize=.75\hsize}X}

\sectionfont{\sffamily\upshape\large}
\subsectionfont{\sffamily\upshape\normalsize}
\subsubsectionfont{\sffamily\mdseries\upshape\normalsize}
\usepackage[authoryear]{natbib}
\usepackage{csquotes}

\renewcommand{\P}{\mathbb{P}}
\renewcommand{\E}{\mathbb{E}}

\newcommand{\betav}{\boldsymbol{\beta}}

\newcommand{\lambdav}{\boldsymbol{\lambda}}
\newcommand{\yv}{\mathbf{y}}

\newcommand{\Xm}{\mathbf{X}}

\newcommand{\thetav}{\boldsymbol{\theta}}

\newcommand{\RR}{\mathbb{R}}

\newcommand{\yrep}{\mathbf{y}_{\mathrm{rep}}}

\renewcommand{\b}{\mathrm{base}}

\newcommand{\Thetav}{\boldsymbol{\Theta}}

\newcommand{\postp}{\mathsf{post-}p}
\newcommand{\partp}{\mathsf{part-}p}
\newcommand{\condp}{\mathsf{cond-}p}
\newcommand{\jointp}{\mathsf{joint-}p}
\newcommand{\sampp}{\mathsf{sampled-}p}
\newcommand{\calp}{\mathsf{cal}-p}

\newtheorem{definition}{Definition}

\title{\bf Joint $p$-Values for Higher-Powered Bayesian Model Checking with Frequentist Guarantees}
\author{Collin Cademartori\\Department of Statistics, Columbia University\\ cac2301@columbia.edu}
\date{September 22, 2023}

\begin{document}
\maketitle


\abstract{We introduce a joint posterior $p$-value, an extension of the posterior predictive $p$-value for multiple test statistics, designed to address limitations of existing Bayesian $p$-values in the setting of continuous model expansion. In particular, we show that the posterior predictive $p$-value, as well as its sampled variant, become more conservative as the parameter dimension grows, and we demonstrate the ability of the joint $p$-value to overcome this problem in cases where we can select test statistics that are negatively associated under the posterior. We validate these conclusions with a pair of simulation examples in which the joint $p$-value achieves substantial gains to power with only a modest increase in computational cost.}

\section{Introduction}\label{sec:introduction}

Checking the adequacy of a statistical model is an essential step in almost any applied modeling workflow \citep{BayesWorkflow1,BayesWorkflow2,BayesWorkflow3,BoxsLoop}. When a model's assumptions have not been tested against their observable consequences, inferences about unobservable quantities obtained through such models must be interpreted skeptically. However, the process of checking a model is often not straightforward, and it is subject to a number of confusions in the Bayesian setting in particular. For instance, we find substantial disagreement in the literature over questions such as:

\begin{enumerate}
\item Is our goal to subject our model to the strongest possible test of its compatibility with (some feature of) the data, in order to have the best possible chance of rejecting the model? Or, is our goal to generate assessments of fitness which can help us to identify particular improvements we can make to the model? We will term the former the \textit{rejection goal}, which is strongly advocated for by \citet{RobinsEtAl00}. This perspective is explicitly rejected by \citet{GelmanUValues}, who advocates instead for more qualitative methods that enable ``[understanding] in what ways the fitted model departs from the data.'' This view closely tracks with our latter goal, which we will refer to as the \textit{discovery goal}.
\item Do we need to know the frequency properties of our model checking procedures in order to interpret their output? Or can we achieve the relevant goals by using ``purely'' Bayesian calculations? And would using frequency calculations undermine the Bayesian consistency or validity of our analysis? Arguments for the importance of frequency information can found in \citet{RobinsEtAl00} and \citet{BayarriBerger00}, whereas arguments in the opposite direction are given in \citet{GelmanExamples} and \citet{GelmanUValues}.
\end{enumerate}

We first note that the distinction between the goals of model rejection and model discovery can be arbitrarily sharp. An oracle which provides a yes or no answer, for any proposed model, to the question of whether it is a valid description of the true data generating process gives us $100\%$ power against any alternative. But such a binary oracle offers little help in diagnosing the source of the model's inaccuracy or in finding plausible directions for improvement. Combined with Box's famous adage - ``all models are wrong, but some are useful'' - one might argue on this basis that pursuing the discovery goal is the more practical strategy. In the Bayesian setting, this is most commonly achieved by comparing observed data to simulations from the model's posterior predictive distribution - i.e. the model's data distribution $p(\cdot\mid\thetav)$ averaged over the posterior $p(\thetav\mid\yv)$. In numerical form, this leads to the posterior predictive $p$-value, but advocates of the posterior predictive check often recommend qualitative visual checks for their higher density of information \citep{BayesWorkflow3}. In this setting, concerns over frequency properties are either not relevant (in the case of the $p$-value, which can be interpreted directly as a posterior probability) or not well-defined (in the case of visual assessments, where formal decision processes are rarely defined).

However, when we pursue the rejection goal, frequency evaluations become much more relevant. \citet{PostP} showed that the posterior predictive $p$-value has a frequency distribution which is stochastically less variable than uniform (when data $\yv$ are sampled from the model's prior predictive distribution $p(\yv)$). As a consequence, the frequency of a given posterior predictive $p$-value is usually less than its nominal value, and sometimes substantially so. If we test the model by comparing the $p$-value to some threshold, then such tests will be conservative or underpowered compared to a test using the corresponding frequency. Moreover, it has been observed that the size of this power deficit can be quite large in practice \citep{HierarchicalCheck1,HierarchicalCheck2,PostPPharma}.

This paper makes three arguments:
\begin{enumerate}
\item The rejection goal can become practically relevant even in a workflow that takes the discovery goal as its primary concern.
\item We can effectively pursue the rejection goal by testing multiple statistics simultaneously with a joint $p$-value that achieves a balance of computational tractability, finite sample performance, and ability to scale with model complexity. 
\end{enumerate}

In practice, pursuing the discovery goal often entails constructing many models, each designed to improve fitness in response to a check of a previous model. 
However, this process of model multiplication must eventually terminate, at least temporarily, due either to the diminishment of identifiable routes for further improvement or the need to use the model for some downstream task. We therefore want to evaluate the risks of stopping at any given time in the model building process. In particular, we may reasonably wish to judge if our current model is acceptable for some task.

In many cases, the acceptability of a model will be directly related to the difficulty of the tests which it has passed. A model which has been checked only superficially is often untrustworthy. Thus, the evaluation of a potential stopping time is connected not just to outcome but also to the power of a test.
Whereas rejection-oriented tools take power-maximization as a first concern, tools oriented towards the discovery goal may be less helpful towards this end, particularly when we consider stopping in response to such discovery-oriented tools becoming less informative. For instance, while checking fitness to individual features of the data is more useful for diagnosing specific sources of misfit, omnibus tests that aim to check the model more holistically can achieve higher power for rejection purposes.

As a general strategy to obtain higher power for model rejection, we propose computing a posterior predictive $p$-value for a collection of test statistics.

\begin{definition}[Joint $p$-value]
  Let $\mathcal{T} = \left\lbrace T_1,\ldots,T_d\right\rbrace$ be a set of test statistics. Then the joint $p$-value is given by
  \begin{equation}
  \label{eq:joint_pval}
    \jointp_{\mathcal{T}}(\yv) = \P_{p(\yrep\mid\yv)}\left[ T_1(\yrep) > T_1(\yv) \text{ and } T_2(\yrep) > T_2(\yv) \text{ and}\; \cdots\; T_d(\yrep) > T_d(\yv) \right].
  \end{equation}
  We also define a sampled variant of the above, given by
  \begin{equation}
    \label{eq:joint_pval_sampled}
    \jointp_{\mathcal{T}}(\yv,\thetav) = \P_{p(\yrep\mid\thetav)}\left[ T_1(\yrep) > T_1(\yv) \text{ and } T_2(\yrep) > T_2(\yv) \text{ and}\; \cdots\; T_d(\yrep) > T_d(\yv) \right],
  \end{equation}
where $\thetav$ is drawn from the posterior distribution $p(\thetav\mid\yv)$, resulting in a random $p$-value.
\end{definition}
The key idea behind our approach is that testing many statistics at once can substantially increase the difficulty of the model check, which can allow large improvements to power with appropriately set thresholds for rejection. This joint $p$-value can be much easier to compute in practice than the most powerful calibration-based model checks, and it enjoys finite-sample guarantees that many simpler methods cannot provide. Furthermore, as we will see in Section \ref{sec:expansion_arg}, computationally simpler alternatives tend to suffer from degrading performance as model complexity increases. By increasing the number of statistics in our test, the joint $p$-value can often scale up with the complexity of the model.
\subsection{Outline}

This paper is organized as follows. In Section \ref{sec:expansion_arg}, we review Meng's bound establishing the conservativity of the posterior predictive $p$-value and show that the gap in this bound tends to grow when a base model is expanded to a larger model. Section \ref{sec:rejection_tools} reviews existing approaches to model rejection within a Bayesian framework and compares them to our proposed joint $p$-value on the basis of their computational and interpretational properties. We also present a simple extension of Lemma 1 in \citet{PostP} which provides a bound on the frequency of any given joint posterior predictive $p$-value. We verify that this strategy can gain power from the additional information contained in the joint structure of the test statistics by studying this bound under copula models of test statistic dependence in Section \ref{sec:copulas}. Section \ref{sec:experiment} then presents a pair of numerical experiments in which we compare our joint $p$-value to a number of alternatives, demonstrating that our method achieves a practically useful trade-off between interpretability, power, and computational tractability. Finally, Section \ref{sec:discuss} discusses the role for our method in a crowded landscape of model checking tools and considers directions for future work.

\section{Model rejection with $\postp$}\label{sec:expansion_arg}

We now present a systematic argument for why a discovery-first modeling workflow should take the rejection goal seriously and why the usual discovery-focused tools cannot be used for this purpose. A very common technique for Bayesian model rejection is to compare the posterior predictive $p$-value ($\postp$) to some threshold. The classic argument against using this procedure for model rejection - that it is overly conservative - can be formalized using the concept of convex order. For distributions $p,q$, we say that $p$ is less than $q$ in convex order ($p \ll q$) if, for $X\sim p$, $Y\sim q$, and any convex function $\psi$, we have that
\begin{equation}
  \label{eq:convex_order}
  \E\psi(X) \leq \E\psi(Y).
\end{equation}
\citet{PostP} showed that $\postp$ is dominated in convex order by a uniform random variable. To demonstrate this, let $p_T(\yv,\thetav)$ be the $p$-value computed with respect to $p(\yv\mid\thetav)$, i.e.
\begin{equation}
  \label{eq:cond_pval}
  p_T(\yv,\thetav) = \P_{p(\yrep\mid\thetav)}\left[  T(\yrep) \geq T(\yv)\right].
\end{equation}
Then we have that
\begin{equation}
  \label{eq:conservative}
  \E_{p(\yv)}\psi\left( \postp_T(\yv) \right) = \E_{p(\yv)}\psi\left( \E_{p(\thetav\mid\yv)}p_T(\yv,\thetav) \right) \stackrel{(a)}{\leq} \E_{p(\thetav)}\E_{p(\yv\mid\thetav)}\psi\left( p_T(\yv,\thetav) \right) \stackrel{(b)}{=} \E\psi(U),
\end{equation}
where $U$ is a uniform random variable. Here, $(a)$ follows by Jensen's inequality, and $(b)$ follows from the definition of $p_T(\yv,\thetav)$ and the fact that any $p$-value has a uniform distribution under its assumed sampling distribution (by the probability integral transform). Roughly, this convex ordering means that $\postp$ will tend to have a distribution that is more peaked around $0.5$, and thus it will commonly be true that
\begin{equation}
  \label{eq:frequency_trend}
  f_{T}(\alpha) \stackrel{\mathrm{def}}{=} \P\left( p_T(\yv) \leq \alpha \right) < \alpha
\end{equation}
for sufficiently small values of $\alpha$. We can thus see that for a sufficiently small threshold $p^*$, when \eqref{eq:frequency_trend} holds, the test that rejects when $p_T(\yv) < p^*$ is lower power than the test that rejects when $f_T\left( p_T(\yv) \right) < p^*$.

The Bayesian who does not want to be concerned with frequency calculations may reasonably wonder at this point whether this claimed power deficit will be an issue in practice. Indeed, this argument does not show that $\postp$ is useless for model rejection. Meng also showed that
\begin{equation}
  \label{eq:frequency_bound_meng}
  \P\left( p_T(\yv) \leq \alpha \right) \leq 2\alpha
\end{equation}
for all $\alpha$. Thus, when $p_T(\yv)$ is sufficiently small, we will still have sufficient information to reject the model on frequentist grounds without the need to compute or approximate $f_T(p_T(\yv))$. Indeed, many examples show that $\postp$ can work quite well for this purpose, and one can always choose a more skeptical threshold if power is a substantial concern.

Of course, the viability of this strategy relies entirely on \textit{how} non-uniform $\postp$ is in any given case. If $\postp$ becomes severely non-uniform and is sharply peaked around $0.5$, then the only way to achieve significance levels that aren't extremely conservative may be to place the nominal threshold at levels so large (e.g. $>0.4$) that they would never be recommended absent direct evidence of this degree of peakedness (since they would result in unreasonably large significance levels in other cases). Consistent with this, it has been observed that large variation in the conservativity of $\postp$ across models and test quantities undermines its consistent interpretation \citep{ModelCheckingGuide}. In short, we can expect $\postp$ to give reasonable rejection performance only when it is consistently not-too-severely non-uniform.

\subsection{Conservativity of $\boldsymbol{\postp}$ and discovery-driven model expansion}

In light of the above arguments, it is clear that we need an understanding of how non-uniform $\postp$ may be in practice to adjudicate the relevant concerns. Examining \eqref{eq:conservative}, we can see that the degree of non-uniformity is entirely controlled by the size of the gap in the inequality $(a)$. It is well-known that the gap in Jensen's inequality can be bounded above and below as
\begin{equation}
  \label{eq:jensen_gap}
  \sigma^2_{\yv}\frac{\inf \psi''}{2} \leq \E_{p(\thetav\mid \yv)}\psi\left( p_T(\yv,\thetav) \right) - \psi\left( \E_{p(\thetav\mid\yv)}p_T(\yv,\thetav) \right)\leq \sigma^2_{\yv}\frac{\sup\psi''}{2},
\end{equation}
where
\begin{equation}
  \label{eq:gap_variance}
  \sigma^2_{\yv} = \mathrm{Var}\left[ p_T(\yv,\thetav)\mid\yv \right].
\end{equation}
Thus, the non-uniformity of $\postp_T(\yv)$ is controlled by the average size of $\sigma^2_{\yv}$. In particular, taking $\psi(x) = (x-1/2)^2$, \eqref{eq:jensen_gap} shows that larger values of $\sigma^2_{\yv}$ imply sharper concentration of the distribution of $\postp_T(\yv)$ around $1/2$. We claim that, for at least some $T$, we should expect $\sigma^2_{\yv}$ to increase throughout a discovery-driven modeling workflow. To formalize this claim, we begin with the following assumption.

\textbf{Workflow Assumption.} \textit{In a modeling workflow that emphasizes an open-ended process of model criticism and model improvement, our models will tend to become more complex and require higher-dimensional parameter spaces in order to accommodate those features of the data which are observed empirically but are not accounted for in our existing models.}

This assumption of model improvement as requiring model expansion may not always hold, for instance if we move from a generic initial model to a more specialized model designed with particular domain knowledge. Nevertheless, we believe this assumption is valid in many settings, as model improvement often requires accounting for unanticipated sources of variation (e.g. overdispersion, random effects, nonlinearity), which results in models that are higher-dimensional than their predecessors. We now formalize the notion of model expansion so that we can study its effects on the variance $\sigma^2_{\yv}$.

\begin{definition}[Model Expansion]
  A model $p(\yv,\thetav,\lambdav)$ defined with additional parameter $\lambdav\in\overline\RR^k$ is an expansion of base model $p_{\b}(\yv,\thetav)$ if
\begin{equation}
  \label{eq:expansion_def}
  p_{\b}(\yv,\thetav) = p(\yv,\thetav\mid\lambdav_0) \text{ for some } \lambdav_0\in\overline{\RR}^k,
\end{equation}

where $\overline{\RR}=[-\infty,\infty]$. If $\lambdav_0 = \pm \infty$, then the right-hand side is interpreted as a pointwise limit as $\lambdav_0\to\pm\infty$.
\end{definition}

In words, $p$ is an expansion of $p_{\b}$ if it embeds $p_{\b}$ as a conditional distribution. Our workflow assumption can be formalized as the proposition that a discovery-driven modeling workflow will tend to produce models which are expansions of previous models. Furthermore, when passing from a base model to an expanded model in this way, we can see by the law of total variance that
\begin{align*}
  \label{eq:p_var_ineq}
  \mathrm{Var}_{p}\left[ p_T(\yv,(\thetav,\lambdav))\mid\yv \right] &= \E\left\lbrace\mathrm{Var}_p\left[ p_T(\yv,(\thetav,\lambdav))\mid\yv,\lambdav \right]\mid \yv\right\rbrace + \mathrm{Var}_p\left\lbrace \E\left[ p_T(\yv,(\thetav,\lambdav))\mid\yv,\lambdav \right] \mid \yv\right\rbrace\\
  &= \mathrm{Var}_{p_{\b}}\left[ p_T(\yv,\thetav) \right] + \Delta + \mathrm{Var}_p\left\lbrace \E\left[ p_T(\yv,(\thetav,\lambdav))\mid\yv,\lambdav \right]\mid\yv\right\rbrace,
\end{align*}
where we define
\[
 \Delta = \E\left\lbrace\mathrm{Var}_p\left[ p_T(\yv,(\thetav,\lambdav))\mid\yv,\lambdav \right]\mid \yv\right\rbrace - \mathrm{Var}_p\left[ p_T(\yv,(\thetav,\lambdav))\mid\yv,\lambdav=\lambdav_0 \right].
\]
We note that the second equality follows from the fact that $p_{\b}(\thetav\mid \yv) = p(\thetav\mid \yv,\lambdav=\lambdav_0)$ and $p_{\b}(\yv\mid\thetav) = p(\yv\mid\thetav,\lambdav=\lambdav_0)$. In any given model expansion, $\Delta$ may be positive or negative, as $\mathrm{Var}_p\left[ p_T(\yv,(\thetav,\lambdav))\mid\yv,\lambdav \right]$ may vary arbitrarily over the support of $p(\lambdav\mid\yv)$. On the other hand, we clearly always have $\mathrm{Var}_p\left\lbrace \E\left[ p_T(\yv,(\thetav,\lambdav))\mid\yv,\lambdav \right]\mid\yv\right\rbrace\geq 0$. Thus, this identity along with \eqref{eq:jensen_gap} strongly suggests that $\sigma^2_{\yv}$ - and thus the non-uniformity of $\postp$ - tends to increase through the process of model expansion.

We also note that this problem is not exclusive to $\postp$. If the posterior predictive $p$-value is highly non-uniform, then we should expect similar posterior predictive checks, such as replication plots, to be problematic for purposes of model rejection as well. A check which produces replications that appear visually similar to the observed data $20\%$ of the time would usually be considered a positive result for the proposed model. It may just as easily be true that if the model were correct, such a visual check would produce sampled data similar to the observed data in a much higher proportion of replications. In short, visual checks can be conservative in the same way as numerical $p$-values.

We draw two conclusions from these observations. First, the above shows that when our tools for model discovery lead us to larger models, they also tend to lead toward models that are harder to reject with observable data insofar as our $\postp$ values become increasingly conservative. While we may be willing to accept a trade-off in favor of tools that emphasize discovery over rejection all else equal, we believe few applied researchers would be comfortable with an arbitrarily high and increasing risk of selecting nearly unfalsifiable models. If this is correct, then this indicates a need to take the rejection goal seriously as an independent concern in model checking.

Second, these calculations show that existing and common model checking tools such as $\postp$ are not suited to the rejection goal at least without some modification. Instead, what is needed is a model checking tool for which the difficulty of the assessment can be scaled to match the complexity of the model appropriately. In the next section, we review some existing proposals for remedying the non-uniformity of $\postp$ and introduce our proposed method, the joint $p$-value.

\section{$p$-Values for model rejection}\label{sec:rejection_tools}

We now consider possible methods for partly remedying the difficulties associated with the posterior predictive $p$-value as a tool for model rejection.  We begin with attempts to derive $p$-values which have exactly or approximately uniform distributions, and then turn to our proposed joint $p$-value.

\subsection{Exactly and approximately calibrated $p$-values}

\citet{PostProcessPValues} propose to overcome the conservativity of $\postp$ by plugging it into (an estimate of) its distribution function, which will result in a uniformly distributed quantity when the model is correctly specified. In particular, if $H$ is the distribution function of $\postp_{\yv}$ when $\yv$ is drawn from the prior predictive distribution, then we can estimate $H$ by the empirical distribution function $\hat{H}(p) = \frac{1}{S}\sum_{s=1}^S\mathbbm{1}\left\lbrace \postp_T(\yrep^{(s)}) \leq p \right\rbrace$, where $\left\lbrace \yrep^{(s)}\right\rbrace_{s=1}^S \stackrel{iid}{\sim} p(\yv)$ is a sample from the prior predictive distribution. The calibrated posterior predictive $p$-value is then
\begin{equation}
  \label{eq:calibrated_pval}
  \calp_{\yv} = \hat{H}\left(\postp_T(\yv)\right).
\end{equation}

This calibration step fully resolves the conservativity problem when $H$ is well-estimated by $\hat{H}$. However, the computation of $\hat{H}$ generally requires sampling from $p(\thetav\mid\yrep^{(s)})$ separately for each $s=1,\ldots,S$. This can quickly become computationally infeasible for moderate $S$ if the model is sufficiently complex.

The sampled posterior predictive $p$-value retains exact uniformity for continuously distributed test statistics $T$, but it is almost universally easy to compute \citep{SampledPValues,PivotalSampled,ChiSquaredGOF}. Unlike the other approaches we consider, this method generates a \textit{random} $p$-value by first drawing a sample $\thetav$ from the posterior distribution $p(\thetav\mid\yv)$ and then computing a $p$-value with respect to $p(\yrep\mid\thetav)$. In symbols:
\begin{equation}
  \label{eq:sampled_pval}
  \sampp_{T}(\yv,\thetav) = \P_{p(\yrep\mid\thetav)}\left( T(\yrep) \geq T(\yv) \right), \quad \thetav\sim p(\thetav\mid\yv).
\end{equation}

The posterior predictive $p$-value is just the expected value of $\sampp_{T}$ over the posterior distribution. Estimating \eqref{eq:sampled_pval} by a Monte Carlo average is generally extremely fast since sampling $p(\yrep\mid\thetav)$ is trivial for most models. 
The uniformity of $\sampp_{T}$ is achieved by not aggregating over the posterior distribution. However, this calibration comes at the cost of increased variance due to the additional randomization of $\thetav$. As a result, the sampled $p$-value can exhibit lower power when this variance is large.


In particular, this property means that the sampled $p$-value, like the posterior predictive $p$-value, generally scales poorly with model expansion. To see this quantitatively, recall that $\postp_T(\yv)$ tends to become more concentrated around $1/2$ with model expansion. Defining $Q(p) = \P_{p(\yv)}\left( \postp_T(\yv) \leq p \right)$, we think of the calibrated $p$-value $Q(\postp_T(\yv))$ as representing the true level of fitness of the model to $T$. The contraction of the distribution of $\postp_T(\yv)$ then implies that the nominal value $\postp_T(\yv)$ must increase with model expansion in order to maintain a fixed true level of fitness $Q(\postp(\yv))$. Then, for any $\delta < \postp(\yv)$, the Paley-Zygmund inequality implies that
\begin{equation}
  \label{eq:paley_zygmund}
  \P\left( \sampp_T(\yv,\thetav) > \delta \right) \geq \phi\left( 4\left[ \postp_T(\yv) - \delta \right] \right),
\end{equation}
where $\phi = \frac{x^2}{1+x^2}$. Because $\phi$ is an increasing function, this implies that our lower bound \eqref{eq:paley_zygmund} - and thus the risk of observing large sampled $p$-values - will also tend to increase with model expansion if the true level of fitness is nondecreasing.

Other methods have been proposed that trade exact calibration for approximate calibration and improved computability (compared to $\calp$) or power (compared to $\sampp$). \citet{BayarriBerger99} propose $p$-values which are Bayesian in the sense that they account for posterior uncertainty but which enjoy reduced conservativity relative to $\postp$ by having a uniform frequency distribution in appropriate asymptotics. The key idea for achieving asymptotic uniformity comes from the observation that $\postp$ involves a double use of the data whereby the posterior ``sees'' the statistic $T$ against which it will subsequently be tested. This artificially reduces the difficulty of the test, leading to conservativity.

This diagnosis is partly justified by considering tests with ancillary statistics $T$. Since these have distributions which are independent of $\thetav$, the posterior contains no information about $T$, and $\postp$ becomes exactly uniform for such $T$. The proposed $p$-values attempt to formalize the idea of ``removing'' the information in $T$ from the posterior before testing. The first of these is the conditional predictive $p$-value, defined for a test statistic $T$ as
\begin{equation}
  \label{eq:conditional_pval}
  \mathsf{cond-}p_T(\yv) = \P_{p\left( \yrep\mid \hat{\thetav}_T \right)}\left( T(\yrep) \geq T(\yv) \right),
\end{equation}
where we define $\hat{\thetav}_T = \mathrm{arg}\max p\left( \yv\mid\thetav,T(\yv) \right)$ as the $T$-conditional maximum likelihood estimate of $\thetav$, and
\begin{equation}
  \label{eq:mle_post}
  p\left( \yrep \mid \hat{\thetav}_T \right) = \int p\left( \yrep\mid \thetav,T(\yv) \right)p\left( \thetav\mid \hat{\thetav}_T \right) d\thetav.
\end{equation}

The key idea in this definition is that $\hat{\thetav}_T$ should capture as much of the information about $\thetav$ contained in the data as possible while excluding the information in $T$. When $\hat{\thetav}_T$ is sufficient for $\thetav$, $\mathrm{cond-}p_T$ is exactly uniform. However, forming and conditioning on the conditional MLE is often computationally difficult. We can instead try to remove the information contained in $T$ from the posterior directly by conditioning $T$ out of the likelihood. This results in Bayarri and Berger's partial predictive $p$-value:
\begin{equation}
  \label{eq:partial_pval}
  \mathrm{part-}p_T\left( \yv \right) = \P_{p\left( \yrep\mid\yv\setminus T(\yv) \right)}\left( T(\yrep) > T(\yv) \right),
\end{equation}
where we define the partial posterior and posterior predictive distributions as
\begin{equation}
  \label{eq:partial_posterior}
  p(\thetav\mid\yv\setminus T(\yv)) \propto p(\yv\mid \thetav, T(\yv))p(\thetav), \quad p(\yrep\mid\yv\setminus T(\yv)) = \int p(\yrep\mid\thetav) p(\thetav\mid\yv\setminus T(\yv)) d\thetav.
\end{equation}
Since $T(\yv)$ is determined exactly by $\yv$, the partial posterior differs from the posterior by a factor proportional to $p(T(\yv)\mid \thetav)^{-1}$.

That these $p$-values approximately succeed in removing the conservativity problem is justified by Theorem 2 of \citet{RobinsEtAl00}, which implies that $\condp$ and $\partp$ both have asymptotically uniform frequency distributions under sampling models of the form
\begin{equation}
  \label{eq:asymp_model}
  p\left( \yv \mid \thetav, \psi_n \right) = \prod_{i=1}^n p_i\left( \yv_i\mid \thetav, \psi_n \right),
\end{equation}
where $\psi_n\in\RR$ is a one-dimensional nuisance parameter. Robins et al. also propose a number of other methods for deriving approximately calibrated $p$-values which depend on either modifications of the test statistic $T$ or on approximate recalibrations of simpler $p$-values such as $\postp$. We do not treat these approaches in detail here since any generally available computational speedups relative to $\partp$ and $\condp$ are usually achieved by exploiting some aspect of the asymptotics of \eqref{eq:asymp_model}, which we argue in the next section is an overly limiting model in many cases. The interested reader can consult \cite{RobinsEtAl00} for details. Because $\partp$ and $\condp$ have identical asymptotic performance under $\eqref{eq:asymp_model}$ and $\partp$ is generally easier to compute, we will focus all subsequent comparisons on $\partp$.

We note that $\partp$ can still suffer from substantial computational costs when $p(T(\yv)\mid \thetav)^{-1}$ is not analytically available, which is usually the case when the model is sufficiently complex. We are unaware of a scheme for estimating this quantity in general other than estimating $p(T(\yv)\mid \thetav)$ with a kernel density estimator and inverting the result (which is the recommended strategy in \citet{BayarriBerger00}). Such kernel density estimates can be highly inefficient in the tails of the density, leading to explosive errors in the inverse.


\subsection{Modifying the test quantity and joint $p$-values}

The exactly and approximately calibrated $p$-values of the last section were based on the idea that posterior predictive checks can be too easy when we fail to set our thresholds for rejection relative to the corresponding frequency distribution. Calibrating the $p$-values allows us to set thresholds appropriately to maintain a certain level of difficulty. However, Meng's bound \eqref{eq:frequency_bound_meng} tells us that the miscalibration problem is asymmetric. If our nominal $p$-value is so small that twice that value is below our threshold, then we can still confidently reject our model on frequentist grounds.

We may also try to increase the difficulty of our tests by modifying our choice of test quantity. One way to achieve this is with ancillary statistics, which yield posterior predictive $p$-values that are exactly uniform \citep{GelmanExamples}.  However, discovering ancillary (or approximately ancillary) statistics is often difficult. And if our workflow assumption above holds, then we expect the discovery of ancillary statistics to become more difficult as our model size increases and our sampling distributions accommodate a greater variety of data behaviors. Since our primary concern is constructing tests for model rejection which can scale with model complexity, this is particularly worrying.

Similarly, the use of pivotal discrepancy measures (which may depend on parameters $\thetav$ as well as data $\yv$) has been proposed since calibrating the corresponding $p$-values is easier \citep{PostDiscrep,PivotalCal,PivotalChecks}. But pivotal quantities may not exist when the observed data are not independent given the parameters, and there are no guarantees that pivotal quantities exist which quantify any particular feature of interest even when this assumption holds.

Another method for increasing the difficulty of model checks is to hold out some portion of the data with which the test quantity is computed and then compare this quantity to the model fit to the remainder of the data \citep{LOOCV,CVIS}. Like the exactly calibrated $p$-value, this approach requires repeated sampling from the posterior distribution for different sets of observed data, which is often prohibitively computationally expensive. Faster approximate procedures have been proposed, but none of these can be applied across all types of models reliably \citep{CVIIS,CVIIS_DM,CVGhosting}.

A more general and easily-applied approach for generating harder tests of our models is obtained by using many test statistics at once. If $\mathcal{T} = \{T_s\}_{s=1}^d$ is a collection of test statistics, then the corresponding joint posterior predictive $p$-value is
\begin{equation}
  \label{eq:joint_pval_2}
  \jointp_{\mathcal{T}}(\yv) = \P_{p(\yrep\mid\yv)}\left( T_1(\yrep) > T_1(\yv) \text{ and } T_2(\yrep) > T_2(\yv) \text{ and}\; \cdots\; T_d(\yrep) > T_d(\yv) \right).
\end{equation}

An obvious problem with using a joint $p$-value is that we expect its observed value to shrink towards $0$ as $d$ increases even if the proposed model is correct. Furthermore, the joint $p$-value no longer satisfies Meng's bound \eqref{eq:frequency_bound_meng}. The first step towards making $\jointp$ useful is thus a simple generalization of Meng's bound which applies to multiple test statistics.

\begin{theorem}[Frequency Bound for $\jointp$]\label{thm:frequency_bound_joint}
  For any level $\alpha\in [0,1]$, we have that
  \begin{equation}
    \label{eq:frequency_bound_joint}
    \P_{p(\yv)}\left( \jointp_{\mathcal{T}}(\yv) \leq \alpha \right) \leq \inf_{s\in[\alpha,1]}\frac{\int_{0}^sF(t)dt}{s-\alpha},
  \end{equation}
  where $F$ is the cumulative distribution function of the random variable
  \begin{equation}
    \label{eq:joint_conditional}
    \P_{p(\yrep\mid\thetav)} \left[T_1(\yrep) > T_1(\yv) \text{ and } T_2(\yrep) > T_2(\yv) \text{ and}\; \cdots\; T_d(\yrep) > T_d(\yv)\right] .
  \end{equation}
\end{theorem}

\begin{proof}
  See Appendix \ref{app:bound_proof}
\end{proof}

To directly estimate the cumulative distribution function of $\jointp$ or $\postp$ requires repeated simulation of the posterior predictive distribution $p(\yrep\mid\yv)$ for each draw of $\yv$ from the model. In all but the simplest models, this requires sampling from $p(\thetav\mid\yv)$ for each such $\yv$, which will often be prohibitively expensive. Theorem \ref{thm:frequency_bound_joint} shows that we can bound the cumulative distribution function of $\jointp$ by an optimum involving the cumulative distribution function of \eqref{eq:joint_conditional}, which can be simulated directly and efficiently with draws from $p(\yrep\mid\thetav)$.

\begin{figure}[t]
  \centering
  \includegraphics[scale=0.55]{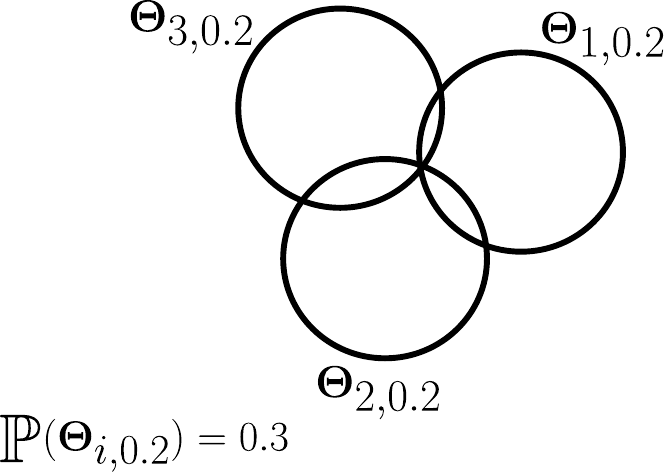}
  \includegraphics[scale=0.55]{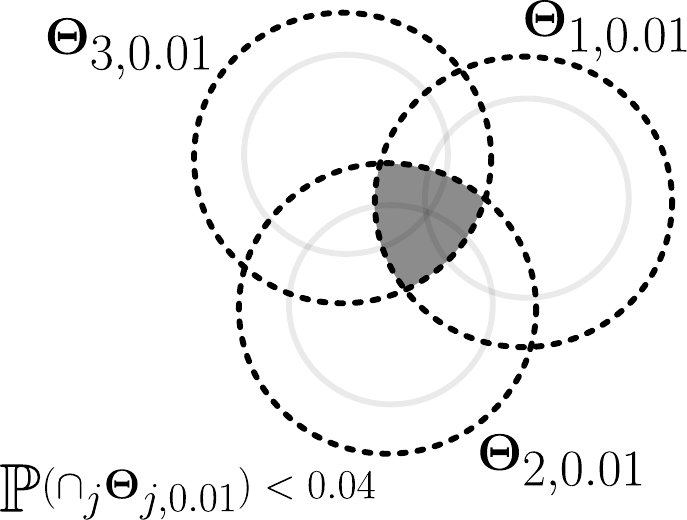}
  \caption{A schematic representation of how the marginal posterior predictive $p$-values can be relatively large while the joint $p$-value is small. In the left panel, because the $\Thetav_{s,0.2}$ have posterior probability $0.3$, $\postp_{T_s}$ is bounded below by $0.3\times 0.2 = 0.06$. In the right panel, because the intersection of the $\Thetav_{j,0.01}$ has posterior probability less than $0.04$, $\jointp$ is bounded above by $0.01\times 0.96 + 1\times 0.04 < 0.05$.}
  \label{fig:joint_schematic}
\end{figure}

Thus, Theorem \ref{thm:frequency_bound_joint} establishes that $\jointp$ can be interpreted for purposes of model rejection by computing a bound on its frequency. But we have yet to establish that $\jointp$ improves on $\postp$ for rejection purposes in general. Since we expect $F$ to increase more sharply at $0$ as $d$ increases, the bound \eqref{eq:frequency_bound_joint} will generally get worse with increasing $d$ for a fixed level of the joint $p$-value. Nevertheless, this bound can still provide value over $\postp$ for rejection purposes if the nominal $p$-value \eqref{eq:joint_pval_2} compensates by falling fast enough with $d$.

This can occur, for instance, when the values of $\thetav$ for which $p(\yv\mid\thetav)$ best fits each $T_s\in\mathcal{T}$ lie in mostly distinct subsets of the parameter space. In particular, define
\begin{equation}
  \label{eq:typical_subspace}
  \Thetav_{s,\alpha} = \{\thetav\in \Thetav \mid \P_{p(\yrep\mid\thetav)}(T_s(\yrep) \geq T_s(\yv)) \geq\alpha \}.
\end{equation}
Each $\Thetav_{s,\alpha}$ can be thought of as the subspace corresponding to data generating processes for which the observed $T_s$ is not atypical. If the $\Thetav_{s,\alpha}$ each have sufficient posterior probability for moderate $\alpha$, then the corresponding $\postp$ values will be too large to reject. Nevertheless, if the $\Thetav_{s,\alpha}$ also have small overlap, then the nominal value of $\jointp$ can be vanishingly small. In such a case, the bound \eqref{eq:frequency_bound_joint} may still be sufficient to reveal the lack of fit. This situation is illustrated in Figure \ref{fig:joint_schematic}.

It is also instructive to compare the joint $p$-value to the sampled $p$-value. Conceptually, both $p$-values exploit the variation in model fit to statistic(s) $T$ across values of the parameter $\thetav$ to increase the power of the resulting test. While the joint $p$-value benefits from additional information (in the form of additional statistics), the bound \eqref{eq:frequency_bound_joint} may be conservative compared to the exact calibration of the sampled $p$-value. Consequently, the sampled $p$-value may have higher power than the joint $p$-value for some problems.

In such cases, we can profitably combine these approaches, obtaining the sampled joint $p$-value $\jointp(\yv,\thetav)$ defined in \eqref{eq:joint_pval_sampled}. The sampled joint $p$-value is computed with respect to $p(\yrep\mid\thetav)$ for some $\thetav\sim p(\thetav\mid\yv)$. When $\yv$ is drawn from the prior predictive distribution $p(\yv)$, the cumulative distribution function of the sampled joint $p$-value is exactly the $F(t)$ in the bound \eqref{eq:frequency_bound_joint}. Since estimating $F$ represents the majority of the computational burden in calculating this bound, the sampled joint $p$-value is generally no faster to compute than the bound on the ordinary joint $p$-value. However, by exploiting the skew of the distribution of \eqref{eq:joint_conditional} and obviating the need to optimize the integral in \eqref{eq:frequency_bound_joint}, the sampled joint $p$-value can have substantially higher power than the joint (and sampled) $p$-value alone. We explore such a case in the second example of Section \ref{sec:experiment}.
\subsection{Computation and interpretation of $p$-values}

We now turn to a comparison of $\calp$, $\partp$, $\sampp$, and $\jointp$ in terms of ease of use and interpretive power for model rejection. For our comparison of computational difficulty, we focus only on $\partp$ and $\jointp$ since $\sampp$ generally poses no computational challenges and $\calp$ is usually computationally infeasible in practice.

\subsubsection{Computing $\partp$ and $\jointp$}

The nominal value of $\jointp$ can be estimated for any $\mathcal{T}$ in the same manner as $\postp$ is estimated for a single statistic. Because $\jointp$ will concentrate near $0$ as $d$ increases, the estimation of $\jointp$ may require a greater number of simulations from $p(\yrep\mid\yv)$ in order to resolve the order of magnitude to acceptable accuracy. However, in practice it often suffices to retain a fixed number of posterior draws $\thetav_i$ and take multiple draws from $p(\yv\mid\thetav_i)$ for each $i$. Thus, the increase in computational overhead from this step is usually modest.

The estimation of the corresponding frequency bound \eqref{eq:frequency_bound_joint} is more taxing, as we will usually not know the cumulative distribution function $F$ in closed form. However, this function can be estimated with inexpensive Monte Carlo simulations of the joint model. Algorithm \ref{alg:joint_p} describes the procedure, repeatedly estimating the empirical CDF of the random variable \eqref{eq:joint_conditional} conditional on $\thetav$ and then aggregating the results. In particular, we only sample from the prior and sampling distributions, never from the posterior, significantly speeding up the Monte Carlo operations compared to exact calibration.
Once we have our estimate $\hat{F}$, we can estimate the bound \eqref{eq:frequency_bound_joint} using quadrature and a grid search for the optimum of $\int_0^s\hat{F}(t)dt/ (s-\alpha)$. Since $\hat{F}$ is a one-dimensional function and $[\alpha,1]$ is compact, this last step can generally be performed very quickly.

\begin{figure}[t]
  \centering
  \begin{algorithm}[H]
    \small
    \caption{Joint $p$-value empirical CDF estimator}\label{alg:joint_p}
    \begin{algorithmic}[1]
      \Require{
        Observed data $\yv$, test statistics $\{T_s\}_{s=1}^S$, \# of prior samples $N_{\mathrm{prior}}$, \# of sampling distribution samples $M_{\mathrm{sampling}}$, \# of samples at which to estimate CDF $L_{\mathrm{estimate}}$.
      }
      
      \For{$n \gets 1,\ldots,N_{\mathrm{prior}}$}
      \State Sample $\{\thetav^{(n)}\} \sim p\left( \thetav \right)$.

      \State Sample $\left\lbrace\yrep^{(m)}\right\rbrace_{m=1}^{M_{\mathrm{sampling}}} \stackrel{iid}{\sim} p\left( \yrep\mid \thetav^{(n)} \right)$.
      
      \For{$l \gets 1,\ldots L_{\mathrm{estimate}}$}

      \State Compute $\hat{p}^{(n,l)} = \frac{1}{M_{\mathrm{sampling}}}\sum_{m=1}^{M_{\mathrm{sampling}}}\mathbbm{1}\left\lbrace T_s\left( \yrep^{(m)} \right) \geq T_s\left( \yrep^{(l)} \right) \text{ for all } 1\leq s\leq S\right\rbrace$
      
      \EndFor

      \State Compute $\hat{F}^{(n)} = q\mapsto \frac{1}{L_{\mathrm{estimate}}}\sum_{l=1}^{L_{\mathrm{estimate}}}\mathbbm{1}\left\lbrace \hat{p}^{(n,l)} \leq q \right\rbrace$.
      
      \EndFor

      \State Compute $\hat{F} = q \mapsto \frac{1}{N_{\mathrm{prior}}}\sum_{n=1}^{N_{\mathrm{prior}}}\hat{F}^{(n)}(q)$
      
      \State \Return $\hat{F}$.
    \end{algorithmic}
  \end{algorithm}
\end{figure}

 When the nominal observed value of $\jointp$ is small, we will need a high resolution estimate of $F$ near $0$ in order to accurately estimate the optimum, and this can require large values of $M_{\mathrm{sampling}}$ and $L_{\mathrm{estimate}}$. Because the complexity of the algorithm scales as
\[
N_{\mathrm{prior}}\times M_{\mathrm{sampling}}\times L_{\mathrm{estimate}},
\]
the cost of estimating $F$ will almost always dominate the computation. In practice, this cost can be substantially reduced by taking advantage of the fact that the computation can be carried out in parallel over the samples $\{\thetav^{(n)}\}$.

The reader may also notice that the $\hat{p}^{(n,l)}$ in Algorithm \ref{alg:joint_p} are not independent across $1\leq l\leq L_{\mathrm{estimate}}$. While this does introduce correlation in the errors, the estimator remains asymptotically unbiased. In particular, \citet{BarbeKendall} showed under weak regularity conditions that the $\sqrt{L_{\mathrm{estimate}}}\left(\hat{F}^{(n)} - F^{(n)}\right)$ converge in distribution to centered Gaussian processes, where $F^{(n)}$ is the CDF of \eqref{eq:joint_conditional} conditional on $\thetav_n$.

The computation of $\partp$ is simpler but more subtle. In all but the simplest cases, we must sample from the partial posterior predictive distribution $p(\yrep \mid \yv\setminus T(\yv))$ in order to estimate $\partp$. This will usually be achieved by sampling first from the partial posterior $p(\thetav\mid\yv\setminus T(\yv))$, which can be done either through direct simulation or by importance resampling draws from the total posterior with the unnormalized weights $1/p(T(\yv)\mid \thetav)$. Whatever our strategy for sampling the partial posterior, we will generally need an estimate of $p(T(\yv)\mid\thetav)$, as this will only be available analytically for the simplest models and test statistics.

In \citet{BayarriBerger00}, it is recommended that kernel density estimation can be applied when the sampling distribution of $T$ is unknown. In theory the required simulation is straightforward, since sampling from $p(T(y)\mid\thetav)$ is as simple as sampling from $p(\yv\mid\thetav)$ and computing $T$ on each sample. Like computing the bound for $\jointp$, this requires a double simulation whereby we first sample $N_{\mathrm{post}}$ values of $\thetav_n\sim p(\thetav\mid\yv)$, and then sample $M_{\mathrm{sampling}}$ values of $\yv_m\sim p(\yv\mid\thetav_n)$ for each $1\leq n\leq N_{\mathrm{post}}$. In practice, we often take $N_{\mathrm{post}}$ much smaller than $N_{\mathrm{prior}}$, but sampling once from $p(\thetav\mid\yv)$ can be much more expensive than sampling once from $p(\thetav,\yv)$.

The greater difficulty in computing $\partp$ is in the need for potentially intractably large values of $M_{\mathrm{sampling}}$. This occurs for instance when the observed value of $T(\yv)$ lies in the tail of the distribution $p(T(\yv)\mid\thetav)$ for some values of $\thetav$ which are probable under the posterior. Because our sampling will generally depend on the inverse of this density, estimating these tails accurately can be essential to avoid explosively large weights. However, kernel density estimation is extremely inefficient in the tails and can systematically underweight tail probabilities with commonly used kernels. Various strategies may be available to stabilize the tail estimation, but we are not aware of any general methods that can succeed reliably without further assumptions or information about the underlying distribution.

\subsubsection{Interpreting $p$-values}

The (approximately) calibrated $p$-values and joint predictive $p$-values face different trade-offs in interpretation. The frequency bound \eqref{eq:frequency_bound_joint} will always be more conservative than the exactly calibrated $p$-value. And as we will see in Section \ref{sec:experiment}, computational intensity tends to trade off with the conservativity of the corresponding $p$-value.  However, the bound \eqref{eq:frequency_bound_joint} holds in total generality and makes no assumptions about asymptotics or exchangeability. We regard this as a substantial benefit of $\jointp$, as the availability of interpretable frequencies is the key property of any model rejection tool.

The asymptotic uniformity of $\partp$ allows it to be interpreted directly (without intermediate bounds) as a frequency in sufficiently nice cases, but this interpretation is limited both by the applicability of the asymptotic model as well as our ability to judge whether we have sufficient data to reliably use the asymptotic approximation. For instance, the asymptotic model \eqref{eq:asymp_model} for $\partp$ assumes both conditional independence as well as a shared parameter vector of fixed dimension. This framework is violated by models parametrized by a vector which grows in dimension with the data (e.g. local parameters in hierarchical models and HMM hidden states), and by models with non-independent sampling distributions (e.g. moving average models).

Furthermore, when our workflow assumption is satisfied, we anticipate that the dimension of the parameter vector will increase as the modeling process proceeds. Consequently, even if the asymptotic assumption appears potentially valid in our initial models, the process of model expansion erodes that validity. Since we were motivated by the problem of finding model rejection tools which are robust in the setting of model expansion, this issue is particularly concerning.

The sampled $p$-value and joint sampled $p$-value both have the beneficial property of being automatically or easily calibrated even preasymptotically. This provides an interpretational advantage over the bound \eqref{eq:frequency_bound_joint}, which can only constrain the underlying probability. However, since both $\sampp(\yv,\thetav)$ and $\jointp(\yv,\thetav)$ are random $p$-values, they are not a direct function of the model and the data, and thus the relationship between any given value of these $p$-values and the true fit of the model to the data is somewhat ambiguous.


\section{Validating $\jointp$ with non-positively associated extremes}\label{sec:copulas}

We now examine the behavior of $\jointp$ under different assumptions about how the exceedance events $\{T_s(\yrep) \geq T_s(\yv)\}$ are associated under the posterior predictive distribution $p(\yrep\mid\yv)$ and the sampling distributions $p(\yrep\mid\thetav)$. We test the behavior of our bound \eqref{eq:frequency_bound_joint} under two conditions: an easier condition in which our test statistics are non-negatively associated under the sampling distributions $p(\yv\mid\thetav)$, which we can study with exact computations, and a harder condition in which the test statistics can be non-positively associated, which we study with simulation experiments using a parametric model for the copula of the statistics.

\subsection{An easier case: non-negatively associated test statistics}

Our main purpose is to establish quantitative evidence under reasonable assumptions for the intuition given after Theorem \ref{thm:frequency_bound_joint}, viz., that our frequency bound \eqref{eq:frequency_bound_joint} will in fact shrink to $0$ as $d\to\infty$ when our extreme exceedances are not positively associated under the posterior predictive and have corresponding marginal $p$-values which are not too large. First, we must establish precisely what we mean by non-positively and non-negatively associated test statistics. To do this, we first generalize our definition of $F$ from Theorem \ref{thm:frequency_bound_joint} to arbitrary random variables. If $(Y_1,\ldots,Y_d)$ are random variables with joint cumulative distribution function $\Phi(y_1,\ldots,y_d)$, then the cumulative distribution function $F_{\Phi}\left( t \right)$ of the random variable $\Phi(Y_1,\ldots,Y_d)$ is the Kendall function associated to the distribution $\Phi$.

If we denote the joint CDFs associated to $(-T_1,\ldots,-T_d)$ under $p(\yv\mid\thetav)$ by $\Phi_{\thetav}$, then $F(t) = \E_{p(\thetav)}F_{\Phi_{\thetav}}(t)$, and we can study the behavior of our frequency bound \eqref{eq:frequency_bound_joint} by studying the Kendall functions $F_{\Phi_{\thetav}}(t)$. (We negate the test statistics in constructing the Kendall function simply to keep the inequality direction consistent with \eqref{eq:frequency_bound_joint}, but nothing of importance is changed since this direction is arbitrary.) Furthermore, if $F_{\Phi_1},F_{\Phi_2}$ are two Kendall functions, then $\Phi_2$ is larger than $\Phi_1$ in positive $K$-dependence order ($\Phi_1 \prec_{\mathrm{PKD}} \Phi_2$) if $F_{\Phi_1}(t) \geq F_{\Phi_2}(t)$ for all $t\in [0,1]$ \citep{PKD}.

To see that this ordering is related to the dependence structure of the distributions $\Phi$, note that if the corresponding random variables $\{Y_i\}_{i=1}^d$ have medians $m_1,\ldots,m_d$ and are strongly positively associated, then there is a stronger probability of all the $Y_i$ lying on the same side of their medians, and $\Phi\left(m_1,\ldots,m_d\right)$ can be relatively large. By contrast, if the $\{Y_i\}_{i=1}^d$ are strongly negatively associated, then we would expect some $Y_i$ to be larger than their medians $m_i$ and some to be smaller, resulting in much smaller $\Phi\left( m_1,\ldots,m_d \right)$. Combining these ideas, we expect $F_{\Phi}(t)$ to increase more rapidly when the $\{Y_i\}_{i=1}^d$ are negatively associated and more slowly when they are positively associated.

This idea can be formalized somewhat by noting that $\Phi_1\prec_{\mathrm{PKD}}\Phi_2$ implies that $\tau_1 \leq \tau_2$, where $\tau_i$ is the value of Kendall's tau associated to $\Phi_i$. In two dimensions, Kendall's $\tau$ is a form of rank correlation and is given by the formula
\[
  \tau = \E\mathrm{sign}\left[\left( Y_1-Y_1' \right)(Y_2-Y_2')\right],
\]
where $(Y_1,Y_2),(Y_1',Y_2')\stackrel{iid}{\sim}\Phi$. This definition can be generalized to higher dimensions, where it measures an overall level and direction of association between the random variables $\{Y_i\}_{i=1}^d$ \citep{JoeConcordance}.

We will say simply that $(Y_1,\ldots,Y_d)\sim\Phi$ are positive $K$-dependent if $F_{\Phi} \prec_{\mathrm{PKD}} F_{\Psi}$, where $\Psi$ is the joint CDF corresponding to independent random variables $U_1,\ldots,U_d$. Because the Kendall function $F_{\Psi}$ is independent of the marginal distributions of the $U_i$, we may take them to be uniform on $[0,1]$. This simplification allows for direct calculation of $F_{\Psi}$, which is shown in \citet{BarbeKendall} to be given by the formula
\begin{equation}
  \label{eq:independent_kendall}
  F_{\Psi}(t) = t\left[ 1 + \sum_{i=1}^{d-1}\frac{\log(1/t)^i}{i!} \right].
\end{equation}

With these facts established, we can now see that if our test statistics are positively associated in the sense of being positive $K$-dependent under $p(\yv\mid\thetav)$ for each $\thetav$, then our frequency bound \eqref{eq:frequency_bound_joint} is further upper bounded by
\begin{equation}\label{eq:independent_bound}
\inf_{s\in[\alpha,1]}\frac{\int_0^s F_{\Psi}(t)dt}{s-\alpha}.
\end{equation}
We study the behavior of this bound under the following assumptions on our posterior predictive $p$-values:
\begin{enumerate}
\item The exceedances for our test statistics $T_s$ are non-positively associated:
  \begin{align}
  \label{eq:nonpos_ppv}
    \P&\left( T_s(\yrep) \geq T_s(\yv) \mid T_1(\yrep) \geq T_1(\yv),\ldots,T_{s-1}(\yrep) \geq T_{s-1}(\yv)  \right)\nonumber\\
    &\hspace{7cm}\leq \P\left( T_s(\yrep) \geq T_s(\yv) \right),
\end{align}
\item The posterior predictive $p$-values are upper bounded by some number $p\in(0,1)$:
  \begin{equation}
    \label{eq:marginal_bound_pval}
    \postp_{T_s}\left( \yv \right) \leq p \text{ for all } 1\leq s\leq d.
  \end{equation}
\end{enumerate}
Under these conditions, we clearly have that $\jointp(\yv) \leq p^d$. In order to assess the performance of $\jointp$ under these assumptions, we plot the relation between \eqref{eq:independent_bound} and the marginal bound $p$ for dimensions $d=2,\ldots,10$ in Figure \ref{fig:test_level_independent}. For each value of the bound $p$, the level monotonically decreases as the number of test statistics increases. This suggests that, in this setting, the joint $p$-value is asymptotically successful in the sense that we will eventually reject the model if we have enough test statistics with non-positively associated observed extremal exceedances.

\begin{figure}[t]
  \centering
  \includegraphics[scale=0.7]{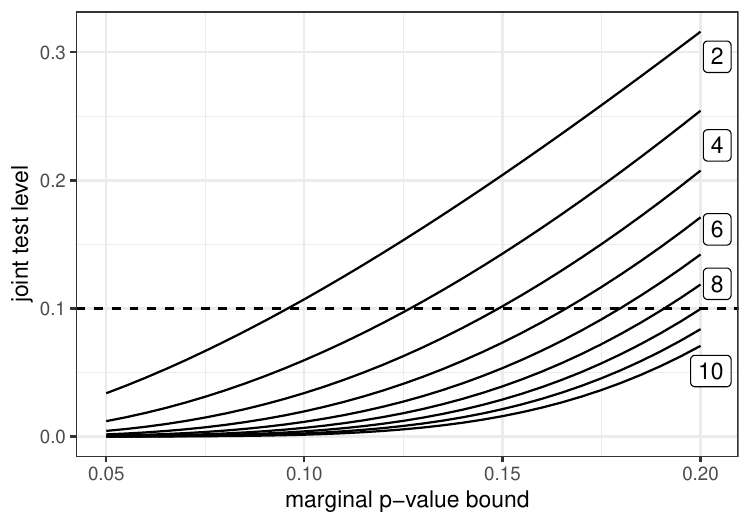}
  \caption{Upper bounds on the frequency of a nominal joint $p$-value $p^d$ for non-negatively associated exceedances versus the bound on the posterior predictive $p$-values of the test statistics separately for various numbers $d$ of test statistics. These curves can be interpreted as lower bounds on the level of a test that would always reject under our assumptions. This level decreases with the number of test statistics and with the bound on the marginal $p$-values.}
  \label{fig:test_level_independent}
\end{figure}

It is also apparent that the efficiency of this procedure depends strongly on the marginal $p$-value bound $p$. For smaller values of $p$, passing from two to three test statistics is sufficient to halve the resulting level of the test, but for larger values, the drop is less than a fifth. This may be particularly troubling since those cases where $p$ is larger are exactly those in which $\postp$ is least capable of model rejection. We note, however, that this represents a worst-case scenario in the sense that we make no assumptions about the gap in the inequalities \eqref{eq:nonpos_ppv}. In practice, when these exceedance events are non-positively associated, we usually observe nominal values of $\jointp$ smaller than $p^d$, and thus we can often obtain bounds that fall below the corresponding curve in Figure \ref{fig:test_level_independent}.

We also note that these conclusions continue to hold if we relax the assumption of positive $K$-dependence for all $\thetav$ to the assumption that
\[
  \E_{p(\thetav)}F_{\Phi_{\thetav}}(t) \leq F_{\Psi}(t) \text{ for } t < \epsilon
\]
for some $\epsilon >0$.
\subsection{A harder case: non-positively associated test statistics}

The results of the last section suggest that our proposed joint $p$-value works when our test statistics are non-negatively associated under the model in the sense that, when our exceedance events $\{T_s(\yrep) \geq T_s(\yv)\}$ are non-positively associated, the resulting frequency bound for $\jointp$ shrinks to $0$ as $d$ grows. The situation for $\jointp$ is harder, however, in situations when the test statistics are non-positively associated under the proposed model. In this case, smaller joint $p$-values are more common under the proposed model, so rejecting the model becomes harder.

Furthermore, investigating the behavior of $\jointp$ in this setting is more challenging since there is no general upper bound available for the average Kendall function $F$. Instead, we use a parametric model of the test statistics $\{T_s\}_{s=1}^d$ to investigate the performance of $\jointp$. In particular, we use a copula model, which specifies only the dependence structure of the statistics while leaving the marginal distributions arbitrary. As noted above, modeling only the dependence structure is possible since the Kendall function is independent of the marginal distributions.

\begin{figure}[t]
  \centering
  \includegraphics[scale=0.7]{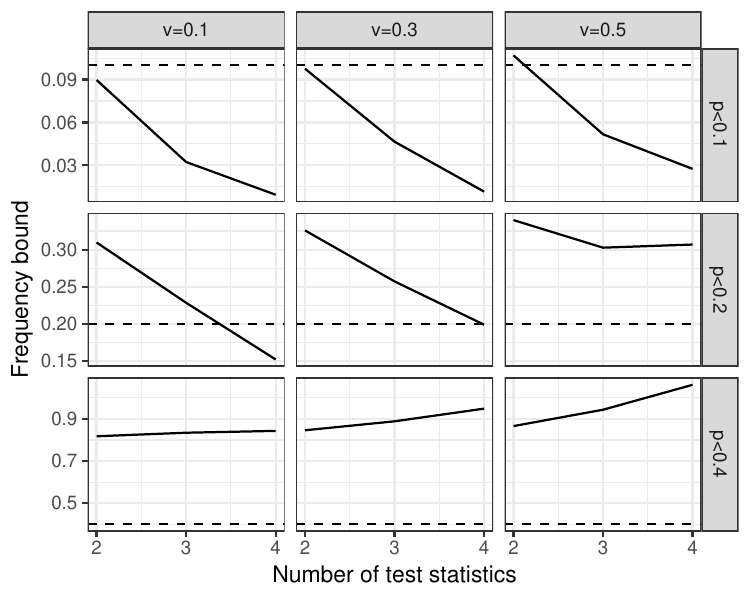}
  \caption{Upper bounds on frequency of joint $p$-value against dimension for negatively associated test statistics and for varying levels of negative dependence (columns) and varying bounds on the marginal posterior predictive $p$-values (rows).}
  \label{fig:bounds_negative}
\end{figure}

We need our copula model to be defined in arbitrarily many dimensions and to allow for modeling negative association between the test statistics, particularly in the tails of their distribution. The Gaussian copula is a natural choice in this setting since it is definable in any dimension and can be parametrized to represent negative association between each component variable. It is also critical that the Gaussian copula has zero tail dependence - a measure of the positive association between component variables when one takes an extreme value - since the behavior of the Kendall function near $0$ is particularly sensitive to these extremes. The $t$ copula, for instance, can be parametrized to represent negative overall association, but the resulting copulas always exhibit positive dependence in the tails \citep{JoeCopulas}. As a result, the corresponding Kendall functions are dominated by $F_{\Psi}$ near zero.

The Gaussian copula $\Phi^G$ is parametrized by a correlation matrix $R$, which we define as
\begin{equation}
  \label{eq:GC_cor}
  R_{ij} =
  \begin{cases}
    1, & i=j\\
    -v/(d-1), & i\neq j,
  \end{cases}
\end{equation}
where $d$ is the number of test statistics. When $v=1$, this is the minimum value that results in a valid correlation matrix $R$ when the off-diagonal entries are constant. In light of the above, we view this as a reasonably hard case for $\jointp$, since we assume that every pair of test statistics is negatively associated in the proposed model and we thus have that $\Phi^G \prec_{\mathrm{PKD}}\Psi$.

There is no known analytic expression for the Kendall function of the Gaussian copula, so we estimate it using the empirical CDF of $\Phi^G\left( U^{(i)}_1,\ldots,U^{(i)}_d \right)$, where $(U_1^{(i)},\ldots,U_d^{(i)})$ are Monte Carlo samples from $\Phi^G$. Plugging this estimate in for $F$ in \eqref{eq:frequency_bound_joint}, we can compute frequency bounds with $\alpha = p^d$ while varying the $\postp$ bound $p$ and the level of negative association $v$. Figure \ref{fig:bounds_negative} plots the resulting bounds against dimension $d=2,3,4$ for $p=0.1,0.2,0.4$ and $v=0.1,0.3,0.5$.

The effectiveness of $\jointp$ in this setting is now contingent on the combination of $p$ and $v$. In the first row, for $p=0.1$, the joint $p$-value continues to work well in the sense that the frequency bound decreases with $d$ fast enough to be significant at the $0.1$ level for just three test statistics and for all tested values of $v$. In the second row, for $p=0.2$, our bound falls with $d$ and is below the corresponding bound for $\postp$ for $v=0.1$ and $v=0.3$, indicating that the joint $p$-value is improving on our ability to reject the model. But for $v=0.5$ the bound actually increases with $d$, indicating that our proposed procedure is no longer able to use the negative association in the observed extremal exceedances to reject the model. When $p$ increases to $0.4$, we find that the frequency bound increases with $d$ regardless of the value of $v$. We conclude that $\jointp$ can still provide value for model rejection when our test statistics are negatively associated under the proposed model, but, for this to be possible and efficient, we need one of the following to hold:
\begin{enumerate}
\item The corresponding marginal $p$-values are not too large.
\item The test statistics are not too negatively associated on average under $p(\yv\mid\thetav)$.
\item The exceedance events $\{T_s(\yrep) \geq T_s(\yv)\}$ are sufficiently negatively associated under $p(\yrep\mid\yv)$ (i.e. the gap in the inequalities \eqref{eq:nonpos_ppv} must be sufficiently large).
\end{enumerate}

\subsection{Copula bounds and the sampled joint $p$-value}

The broad conclusions from these copula-based experiments continue to hold when the joint $p$-value is replaced by the sampled joint $p$-value. Since $\jointp(\yv)$ is just the expected value of the sampled joint $p$-value $\jointp(\yv,\thetav)$ under $p(\thetav\mid\yv)$, the upper bound $p^d$ for $\jointp$ can be interpreted as an upper bound on the average of $\jointp(\yv,\thetav)$. We may still hope that, due to the skewness of the distribution of $\jointp(\yv,\thetav)$ under $\thetav\sim p(\thetav\mid\yv)$, the median of $\jointp(\yv,\thetav)$ over $p(\thetav\mid\yv)$ could tend to zero faster than the mean $\jointp(\yv)$, and could thus achieve better performance after calibration. However, this is only possible if the distribution of $\jointp(\yv,\thetav)$ has a sufficiently heavy right tail under $p(\thetav\mid\yv)$, as measured by the coefficient of variation $\mathsf{cv}=\sigma/\mu$, where $\mu$ and $\sigma$ are the mean and standard deviation of this distribution. If $\mathsf{cv} < C$ for some fixed, finite $C > 0$ for all choices of test statistics, then the Paley-Zygmund inequality implies a nonvanishing probability of obtaining sampled joint $p$-values at least a constant fraction of the mean. In particular, we have for $\delta\in (0,1)$ that
\begin{equation}
  \label{eq:copula_pz}
  \P_{p(\thetav\mid\yv)}\left[ \jointp(\yv,\thetav) \geq (1-\delta)\postp(\yv) \right] \geq \frac{1-\delta}{C+1-\delta}.
\end{equation}
 
Thus, we do not generally expect typical sampled joint $p$-values to tend to zero at a faster exponential rate than the joint $p$-value. To analyze the \textit{calibrated} sampled joint $p$-values, the bounds \eqref{eq:frequency_bound_joint} would be replaced by the evaluated cumulative distribution functions $F(\jointp(\yv,\thetav))$ in Figures \ref{fig:test_level_independent} and \ref{fig:bounds_negative}. These are smaller than the corresponding bounds (since they avoid the optimization step in \eqref{eq:frequency_bound_joint}), but generally only by a factor of about $2$. Thus, while $\jointp(\yv,\thetav)$ is usually smaller than $\jointp(\yv)$, the trends established in this section - particularly the relationship between $p$-value performance and dependence among test statistics - generally apply to $\jointp(\yv,\thetav)$ as well.

\section{Comparing Bayesian $p$-values in simulation examples}\label{sec:experiment}

We now turn to a comparison between $\jointp$, $\partp$, $\sampp$, and $\calp$ in two simple simulation examples. In the first example, we examine a case where the joint $p$-value is more powerful than less computationally expensive alternatives while getting close to the power of approaches which are much more computationally costly. In the second example, we present a case that is challenging for the joint $p$-value and demonstrate that the sampled joint $p$-value can consistently dominate both $\jointp(\yv)$ and $\sampp(\yv,\thetav)$.

\subsection{Testing the quantiles of a beta population}

We take our observed data to be a random sample of size $N=100$ from a $\mathsf{beta}\left( 1,1.5 \right)$ distribution. We will assume a misspecified model with a $\mathsf{beta}\left( \theta,\theta \right)$ sampling distribution and a uniform prior over $[0,3]$. While our assumed sampling distribution is symmetric for all values of $\theta$, the true data generating distribution is substantially skewed to the right, as shown in Figure \ref{fig:beta_dists}. We test our assumed model with statistics $T_1$ and $T_2$ taken to be the $0.05$ and $0.95$ sample quantiles respectively.

\begin{figure}[t]
  \centering
  \includegraphics[scale=0.5]{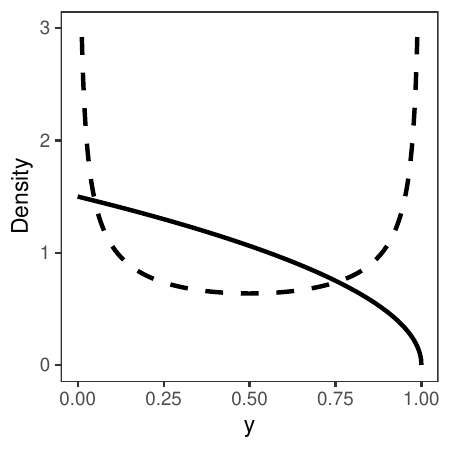}
  \includegraphics[scale=0.5]{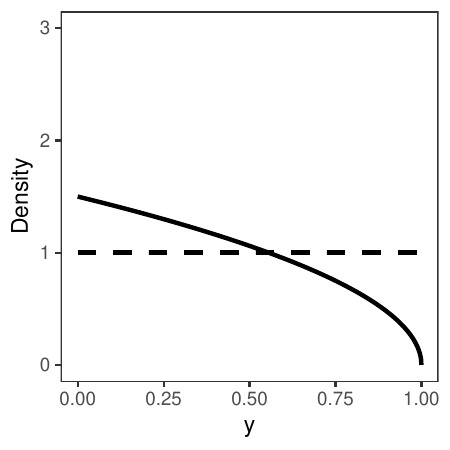}
  \includegraphics[scale=0.5]{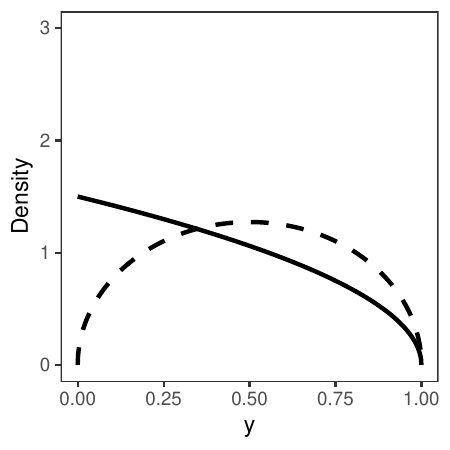}
  \includegraphics[scale=0.5]{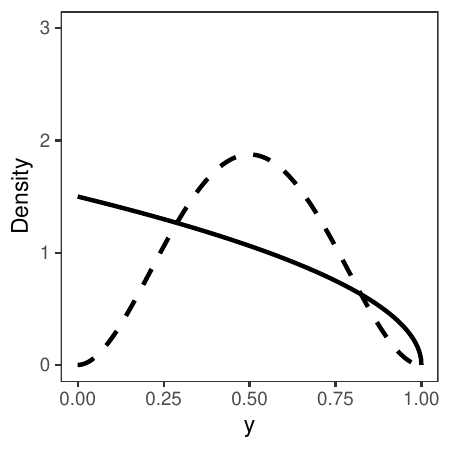}
  \caption{Solid: the $\mathsf{beta}(1,1.5)$ distribution from which our observed data was generated. Dashed: our assumed $\mathsf{beta}(\theta,\theta)$ sampling distribution for $\theta = 0.5,1,1.5,3$.}
  \label{fig:beta_dists}
\end{figure}

Qualitatively, values of $\theta$ closer to zero yield sampling distributions that better match the observed lower quantile but overshoot the observed upper quantile. Similarly, larger values of $\theta$ better match the observed upper quantile but now overshoot the observed lower quantile. If the posterior splits the difference between these regions of parameter space, then we should expect that both observed quantiles will be lower than what is typical in posterior predictive replications from the assumed model. Indeed, this is precisely what we see when we compute the probabilities of $T_s(\yrep)\leq T_s(\yv)$ for $s=1,2$, yielding posterior predictive $p$-values of $\approx 0.07$ for both test statistics.

Computing our bound for $\jointp$ and $\calp$ requires estimating distribution functions of certain exceedance probabilities, and computing $\partp$ requires estimating the partial posteriors \eqref{eq:partial_posterior} for $T_1$ and $T_2$. Figure \ref{fig:joint_sim_check} displays the estimated distribution function $\hat{F}$ around the nominal joint $p$-value $\alpha = 0.0028$ along with the optimization objective $\int_{\alpha}^s \hat{F}(s) / (s-\alpha)$. This shows in particular that we have estimated the distribution $F$ to sufficient resolution around $\alpha$ to trust our estimated bound. Plots of estimated intermediate quantities for $\partp$ and $\calp$ are given in the appendix.

\begin{figure}[t]
  \centering
  \includegraphics[scale=0.4]{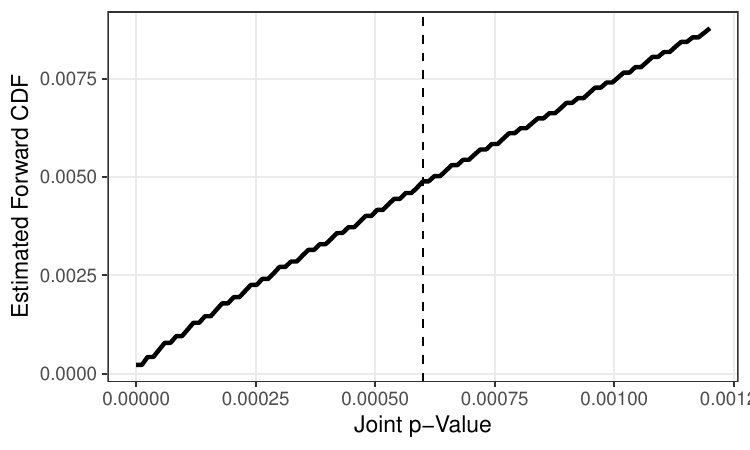}
  \includegraphics[scale=0.4]{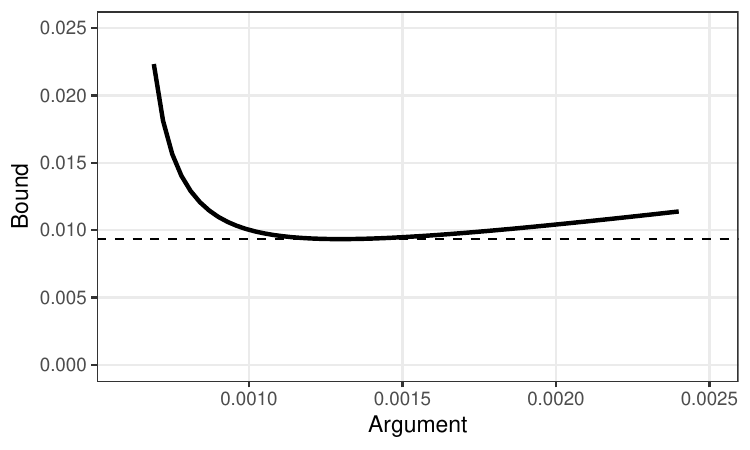}
  \caption{Left: the estimated distribution function of joint extremal exceedances \eqref{eq:joint_conditional} with vertical line indicating the nominal joint $p$-value. Right: the optimization objective on the right-hand side of \eqref{eq:frequency_bound_joint} for a range of $s$ with horizontal line at the optimum.}
  \label{fig:joint_sim_check}
\end{figure}

Table \ref{fig:simulation_results} displays estimates of the various candidate $p$-values for this problem. Because the nominal joint $p$-value is two orders of magnitude smaller than either of the $\postp$ values, we achieve a frequency upper bound which matches the smaller of the two partial $p$-value and calibrated $p$-values. 
Furthermore, while $\jointp$ automatically controls for multiple testing by accounting for dependence, the other results have somewhat artificially higher power due to a lack of multiple testing adjustment. 


\begin{figure}[t]
  \begin{center}
    \begin{tabular}{|c|c|c|c|c|}
      \hline $\postp_{T_s}$ (bound) & $\jointp$ (bound) & $\partp_{T_s}$ & median $\sampp_{T_s}$ & $\calp_{T_s}$\\\hline
      $0.067\; (0.135), 0.027\; (0.055)$ & $0.0006\; (0.009)$ & $0.021,0.009$ & $0.039,0.018$ & $ 0.025, 0.012$ \\\hline
    \end{tabular}
  \end{center}
  \caption{Candidate $p$-values and corresponding bounds (in parentheses, where applicable) for $T_1$ and $T_2$ equal to the $0.05$ and $0.95$ sample quantiles. The partial and calibrated $p$-value give strongest evidence against the model, followed by the joint $p$-value, median sampled $p$-value, and posterior predictive $p$-value.}
  \label{fig:simulation_results}
\end{figure}

Compared to the sampled $p$-value, our frequency bound is less than median sampled $p$-value for either statistic. More importantly, the sampled $p$-value is a random quantity and is larger than our frequency bound in more than $66\%$ of samples for either statistic. Figure \ref{fig:sampled_pvals} displays the survival function of the sampled $p$-values for $T_1$ and $T_2$ along with various other $p$-values, showing that the sampled $p$-value is only less conservative on average compared to the posterior predictive $p$-value. Thus, while the sampled $p$-value is much easier to compute than our bound, it is more conservative in this problem.

\begin{figure}[t]
  \centering
  \includegraphics[scale=0.5]{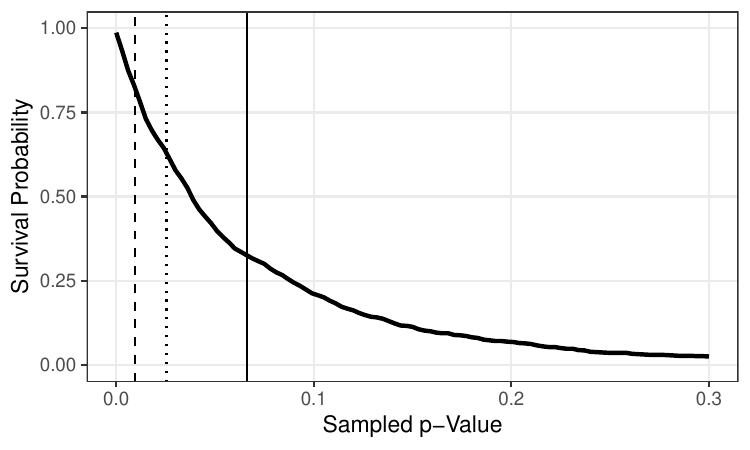}
  \includegraphics[scale=0.5]{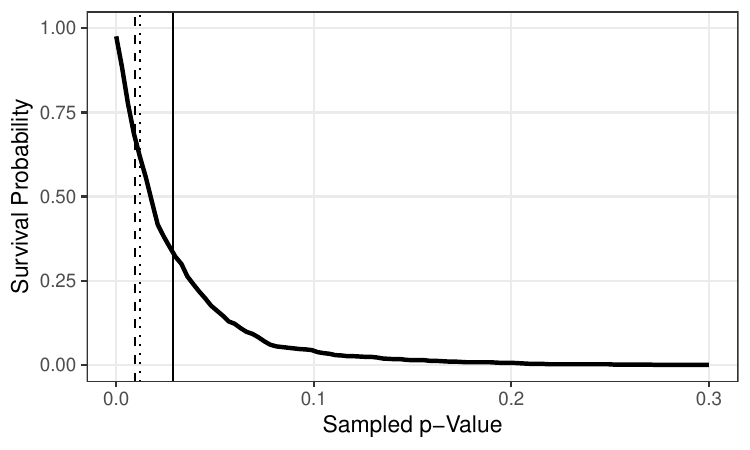}
  \caption{Survival functions for the sampled $p$-value computed for test statistics $T_1$ and $T_2$ equal to the $0.05$ and $0.95$ sample quantiles respectively. The solid, dashed and dotted lines represent the corresponding posterior predictive, joint, and calibrated $p$-values respectively. The sampled $p$-value is less conservative than $\postp$ on average in this problem, but more conservative than $\jointp$ and $\calp$ on average.}
  \label{fig:sampled_pvals}
\end{figure}

Unsurprisingly, at about a sixth the magnitude, the joint $p$-value bound is a substantial improvement over the corresponding frequency bounds for the individual posterior predictive $p$-values. Overall, $\jointp$ displays a useful balance of trade-offs in this problem. It is substantially less conservative than the alternatives which are easier to compute, and it comes within a factor of two of the more powerful alternatives while being easier to compute and offering preasymptotic guarantees.

\subsection{Testing for large effects in a linear regression model}

We now consider a conjugate Bayesian linear regression model of the following form.

\begin{equation}
  \label{eq:reg_setup}
  \yv\mid \betav \sim \mathsf{normal}\left( \Xm\betav,\sigma^2\mathbf{I} \right), \quad \betav\sim\mathsf{normal}\left( \mathbf{0},\boldsymbol{\Sigma} \right),
\end{equation}
where $\sigma^2$ and $\boldsymbol{\Sigma}\in\RR^{d\times d}$ are taken to be fixed hyperparameters for simplicity, and $\Xm\in\RR^{n\times d}$ is known matrix of covariates. We construct the prior covariance as
\begin{equation}
  \label{eq:sigma_def}
  \boldsymbol{\Sigma}_{ij} =
  \begin{cases}
    1, & i=j\\
    \rho, & \{i,j\} = \{1,2\}\\
    0, & \text{otherwise}
  \end{cases}.
\end{equation}

In other words, $\boldsymbol{\Sigma}$ is the identity matrix except for the $(1,2)$ and $(2,1)$ entries, which encode a possibly nonzero prior correlation between $\betav_1$ and $\betav_2$. We test the fit of this model to data generated from the following process.
\begin{equation}
  \label{eq:reg_dgp}
  \yv\mid \betav^* \sim \mathsf{normal}\left( \Xm\betav^*,\sigma^2\mathbf{I} \right),\quad \{\Xm_{ij}\} \stackrel{iid}{\sim} \mathsf{normal}(0,1),
\end{equation}
where $[\betav^*]_j = 4$ if $j=1,2$ and $[\betav^*]_j = 1$ if $j\geq 3$. The model \eqref{eq:reg_setup} is misspecified for this data generating process in the sense that values of $\betav_j$ as large as $4$ are extremely rare under the prior covariance \eqref{eq:sigma_def}.

In our experiment, we generate $n=200$ observtions from $d=100$ covariates. To test the fit of the model \eqref{eq:reg_setup} to datasets generated from the process \eqref{eq:reg_dgp}, we employ test statistics $T_1(\yv) = \Xm_1^T\yv$ and $T_2(\yv) = \Xm_2^T\yv$, where $\Xm_j$ is column $j$ of $\Xm$. In particular, we compare the sampled $p$-value with respect to $T_1$, the joint $p$-value with respect to $\{T_1,T_2\}$, and the sampled joint $p$-value with respect to the same statistics. We are particularly interested in the relative performance of these $p$-values as we vary the correlation hyperparameter $\rho$ over the interval $[0,-0.8]$.

The law of total covariance implies that
\begin{align}
  \label{eq:cov_decomp}
  \mathrm{Cov}\left( T_1(\yrep),T_2(\yrep)\mid \yv \right) &= \E\left[\mathrm{Cov}\left( T_1(\yrep),T_2(\yrep)\mid \betav \right)\mid\yv\right]\nonumber \\
  &\hspace{1cm}+ \mathrm{Cov}\left( \E\left[T_1(\yrep)\mid\betav\right],\E\left[T_2(\yrep)\mid\betav\right]\mid \yv\right)
\end{align}

The first term on the right-hand side affects the slope of the distribution function $F(t)$ in the bound \eqref{eq:frequency_bound_joint} near $0$ (for the reasons discussed in Section \ref{sec:copulas}), whereas the second term affects the magnitude of the ratio of $\jointp_{\{T_1,T_2\}}$ to the $\postp_{T_i}$ (since, e.g., we expect a negative covariance between the conditional expectations to increase the gap in \eqref{eq:nonpos_ppv}). As the computations in Section \ref{sec:copulas} demonstrated, the joint $p$-value bound is most powerful when it can exploit a combination of relatively slowly increasing $F(t)$ and a very small value of $\jointp$, as is most likely to occur when the first covariance on the right-hand side of \eqref{eq:cov_decomp} is positive and the second is negative.

With our model \eqref{eq:reg_setup} and our chosen test statistics, these two covariances are given by $\sigma^2[\Xm^T\Xm]_{12}$ and
\begin{equation}\label{eq:cond_exp_cov}
\sigma^2[\Xm^T\Xm\left( \Xm^T\Xm + \boldsymbol{\Sigma}^{-1} \right)^{-1}\Xm^T\Xm]_{12},
\end{equation}
respectively. When the influence of the prior is not too strong, the latter approaches the former. And in our simulation experiments, these covariances always had the same sign when $\rho=0$, leading to either relatively large values of $\jointp$ or relatively steep distribution functions $F(t)$. This is thus a particularly challenging problem for $\jointp$. Decreasing $\rho$ introduces additional negative correlation between $\betav_1$ and $\betav_2$ in the posterior, which allows for cases in which $[\Xm^T\Xm]_{12}$ is positive but \eqref{eq:cond_exp_cov} is negative and creates more favorable conditions for $\jointp$.

\begin{figure}[h]
  \centering
  \includegraphics[scale=0.66]{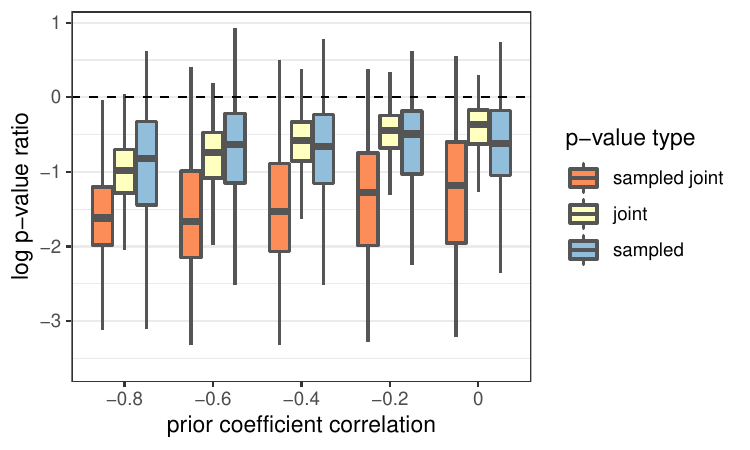}
  \caption{Box plots of the log ratio of candidate $p$-values to twice the posterior predictive $p$-value under repeated sampling of the regression model \eqref{eq:reg_dgp}.}
  \label{fig:reg_ex_comparison}
\end{figure}

Figure \ref{fig:reg_ex_comparison} displays box plots for the log ratio of the joint, sampled, and sampled joint $p$-values to twice the posterior predictive $p$-value as a function of $\rho \in [0,-0.8]$. Normalizing by $2\postp$ in this way adjusts for random variation in the fitness of the model \eqref{eq:reg_setup} to test statistics sampled from \eqref{eq:reg_dgp}, allowing for a direct comparison of relative performance. Values below the dashed line at $0$ indicate $p$-values smaller than $2\postp$. For $\rho$ near $0$, $\jointp(\yv)$ tends to perform worse than $\sampp(\yv,\betav)$ while exhibiting substantially lower variance, reflecting the relative lack of joint posterior information for $\jointp(\yv)$ to exploit. However, as we expect from the above discussion, the performance of $\jointp$ is substantially more sensitive to decreasing $\rho$ than the performance of $\sampp$. Consequently, for $\rho \leq -0.6$, the median $\jointp$ is eventually smaller than the median $\sampp$.

On the other hand, we see that the sampled joint $p$-value $\jointp(\yv,\betav)$ is able to dominate both $\jointp(\yv)$ and $\sampp(\yv,\betav)$ across all values of $\rho$, with a median value at least half an order of magnitude smaller than the median sampled or joint $p$-values in each case. This makes sense, since combining these strategies can attenuate their individual weaknesses. Compared to $\jointp(\yv)$, the sampled joint $p$-value is not as sensitive to skewness in the distribution of \eqref{eq:joint_conditional} and can be calibrated exactly rather than relying on bounds. Compared to $\sampp(\yv,\betav)$, the sampled joint $p$-value can exploit negative posterior predictive dependence between test statistics to obtain more consistently small $p$-values across all values of $\betav$, resulting in reduced variance compared to $\sampp$.

\section{Discussion and future work}\label{sec:discuss}

The limitations of the posterior predictive $p$-value for purposes of model rejection have been established in the literature since at least \citet{PostP} - nearly thirty years prior to the present time of writing. And, as we have discussed, many alternatives with greater power and lesser conservativity have been proposed in the intervening time. Yet, as far as we can tell, none of these has succeeded in establishing itself as widely used, recommended, and implemented. In light of this, the reader might question the productivity of proposing another alternative to $\postp$.

Indeed, we are sympathetic to this viewpoint, and we have therefore aimed to shape our new proposal around a diagnosis of the present state of relative stasis. Disagreement over the purpose of model checking may contribute significantly to this stasis insofar as it limits our ability to form a consensus around any particular set of tools. Some have rejected the sort of model checking that we consider here as fundamentally anti-Bayesian due to the way the data are used to update our beliefs outside of the posterior distribution. Others who see a role for frequentist considerations in Bayesian modeling nevertheless have argued for prioritizing the discovery goal of model checking, deemphasizing the limitations of common tools like $\postp$ for the rejection goal.

A more subtle reason that alternatives to $\postp$ have not been widely adopted may be a tension that limits their usefulness over $\postp$ in practice. In particular, $\condp$ and $\partp$ maximize their power and have known frequency properties only asymptotically. And as we showed, the conservativity of $\postp$ and $\sampp$ increases as model complexity increases, as in process of model expansion. In other words, the existing alternatives to $\postp$ are most useful when $\postp$ is the least problematic. By estimating non-asymptotic frequency bounds (or exact frequencies for the sampled joint $p$-value), our proposal allows consistent frequentist interpretations across modeling contexts. Furthermore, by combining multiple statistics, we have shown that we can scale the power of $\jointp$ with the complexity of the model.

In response to disagreements over the purpose of model checking, we again emphasize that model checking is necessarily a big tent containing distinct goals, no one of which can universally take priority over the others. We view $\jointp$ as a tool specialized to the goal of model rejection, and thus as a complement rather than a substitute for tools already widely in use, which may be better suited for other model checking tasks such as the discovery of alternative models.

The practical applicability of $\jointp$ may be limited when the computations involved become overly burdensome. When the nominal joint $p$-value \eqref{eq:joint_pval} is extremely small, estimating the CDF $F$ in the bound \eqref{eq:frequency_bound_joint} to high accuracy around this nominal value may become very difficult. Particularly troubling is the fact that our Algorithm \ref{alg:joint_p} may spend substantial resources estimating $F$ globally whereas we normally only require an estimate near $0$ to obtain a good upper bound on \eqref{eq:frequency_bound_joint}. Thus, finding a means of reducing this seemingly extraneous computation to increase the efficiency of estimating our bound be a useful direction for future work.

\bibliographystyle{plainnat}
\bibliography{joint_pvals}

\appendix

\section{Proof of Theorem \ref{thm:frequency_bound_joint}}\label{app:bound_proof}

\begin{proof}
  Replacing $\postp$ by $\jointp$, the inequality in \eqref{eq:conservative} shows that $\jointp$ is dominated in convex order by
  \begin{equation}
    \label{eq:conditional_joint}
    \E_{p(\yrep\mid\thetav)}\mathbbm{1}\left\lbrace T_1(\yrep) > T_1(\yv) \text{ and } T_2(\yrep) > T_2(\yv) \text{ and}\; \cdots\; T_d(\yrep) > T_d(\yv) \right\rbrace
  \end{equation}
  under $p(\yv,\thetav)$. However, unlike the case of $\postp$, \eqref{eq:conditional_joint} is no longer uniformly distributed. Nevertheless, the argument of Lemma 1 in \citet{PostP} extends directly. Let $G$ and $F$ be the cumulative distribution functions of $\jointp$ and \eqref{eq:conditional_joint} respectively. Then the established dominance in convex order is equivalent to the inequality
  \begin{equation}
    \label{eq:cumulative_ineq}
    \int_{s}^1 \left[ 1 - G(t) \right]dt \leq \int_s^1\left[ 1 - F(t) \right] dt
  \end{equation}
  for all $s \in [0,1]$. Next note that
  \begin{align}
    \int_0^1 \left[ 1 - G(t) \right] dt &= \E_{p\left( \yv,\yrep \right)}\mathbbm{1}\left\lbrace T_1(\yrep) > T_1(\yv) \text{ and } T_2(\yrep) > T_2(\yv) \text{ and}\; \cdots\; T_d(\yrep) > T_d(\yv) \right\rbrace\nonumber\\
    &= \int_{0}^1 \left[ 1 - F(t) \right] dt.\label{eq:expectation_equiv}
  \end{align}
  Combining \eqref{eq:cumulative_ineq} and \eqref{eq:expectation_equiv}, we get that
  \begin{equation}
    \label{eq:cumulative_ineq_2}
    \int_0^s G(t) dt \leq \int_0^s F(t) dt
  \end{equation}
  for all $s\in [0,\alpha]$. Next, since $G$ is nondecreasing, we can get for any $\alpha\in [0,1]$ the further bound
  \begin{equation}
    \label{eq:G_bound}
    \int_0^s G(t) dt = \int_0^\alpha G(t)dt + \int_\alpha^s G(t) dt \geq \int_0^\alpha G(t)dt + G(\alpha)(s-\alpha).
  \end{equation}
  For $\alpha \in [0,s]$, rearranging and combining with \eqref{eq:cumulative_ineq_2} gives
  \begin{equation}
    \label{eq:desired_bound}
    G(\alpha) \leq \frac{\int_0^sF(t) dt - \int_0^\alpha G(t) dt}{s-\alpha}\leq \frac{\int_0^sF(t) dt}{s-\alpha}.
  \end{equation}
  Optimizing over $s\in [\alpha,1]$ on the right then gives the desired bound.
\end{proof}

\section{Details of sample quantile simulation example}

For the calibrated $p$-values, the empirical CDF of $\postp$ is estimated using $2000$ prior predictive samples $\yrep$, and each $\postp(\yrep)$ is estimated using $1000$ samples from the posterior $p(\theta\mid\yrep)$. The estimated CDF is plotted in a neighborhood of the observed posterior predictive $p$-values in Figure \ref{fig:cal_plot}.

\begin{figure}[h]
  \centering
  \includegraphics[scale=0.5]{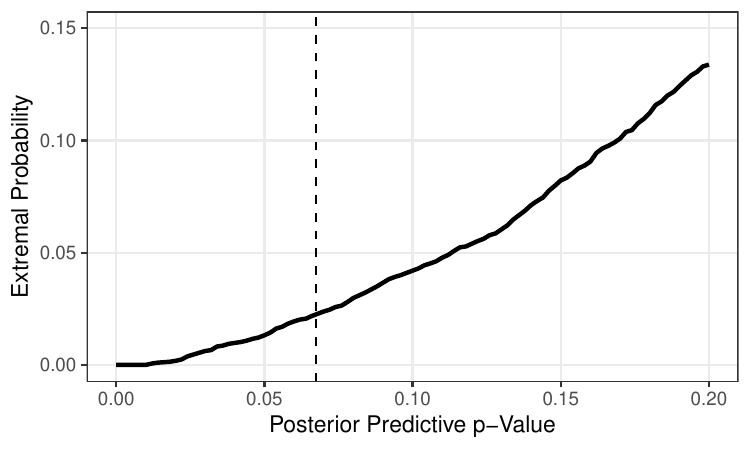}
  \includegraphics[scale=0.5]{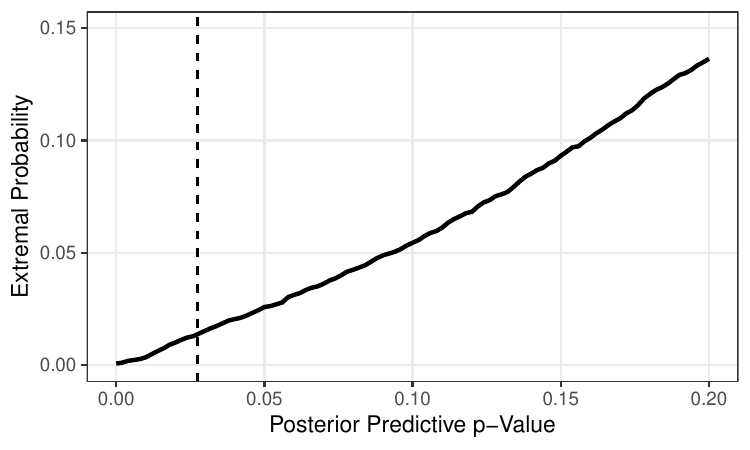}
  \caption{Empirical CDF of $\postp_T(\yrep)$ for $2000$ prior predictive draws of $\yrep$ and for $T$ equal to the $0.05$ (left) and $0.95$ (right) sample quantile statistics, respectively.}
  \label{fig:cal_plot}
\end{figure}

For the partial $p$-values, we calculate the partial posteriors on a grid of $\theta$ values between $0.5$ an $4$ (the support of the prior), estimating the likelihoods of the corresponding test statistics with kernel density estimators. These partial posteriors are displayed in Figure \ref{fig:partial_plot}. The kernel density estimates of the likelihoods for a given value of $\theta$ are computed from $2000000$ samples of the $0.95$ sample quantile and $800000$ samples of the $0.05$ sample quantile. These large sample sizes reflect the need to have low estimation error in the tails of our kernel density estimates to prevent explosive errors when inverting them to form the partial posteriors. As \ref{fig:partial_plot} shows, the estimation accuracy could still be improved further with greater sample sizes.

\begin{figure}[h]
  \centering
  \includegraphics[scale=0.5]{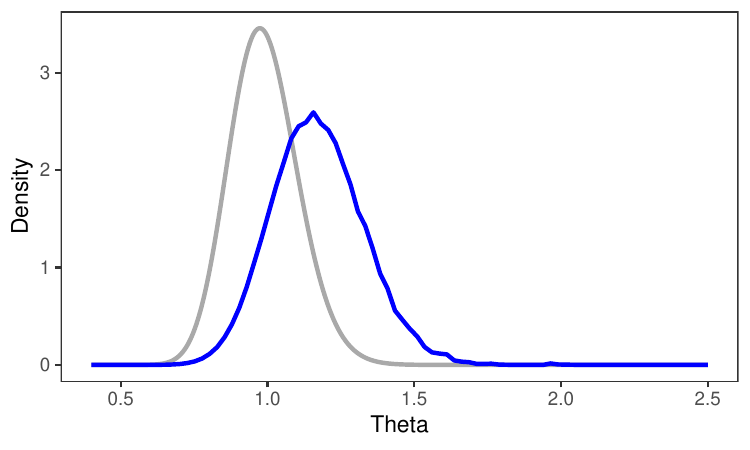}
  \includegraphics[scale=0.5]{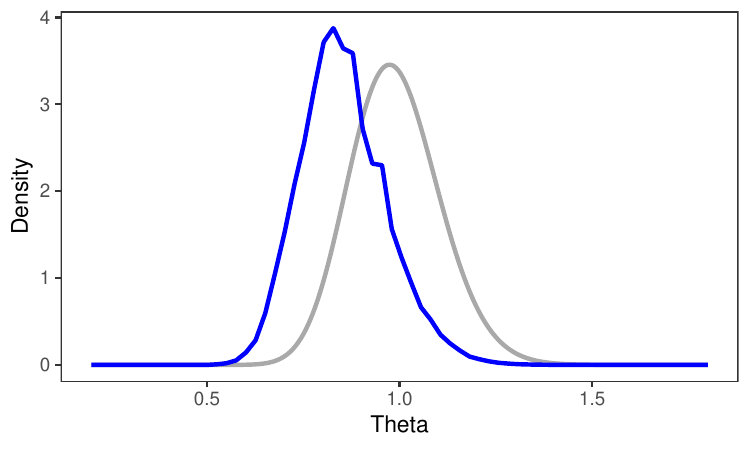}
  \caption{Partial posteriors (in blue) computed for the $0.05$ sample quantile (left) and $0.95$ sample quantile (right), plotted against the full-data posterior (in gray).}
  \label{fig:partial_plot}
\end{figure}

The sampled $p$-values are estimated using $1000$ draws from the sampling distribution $\mathsf{beta}(\theta,\theta)$ for each $\theta$ sampled from the posterior. The joint $p$-value bound is computed using $N_{\mathrm{prior}}=250$ samples $\theta_n$ from the prior distribution and $M_{\mathrm{sampling}}=50000$ samples from each corresponding sampling distribution $\mathsf{beta}(theta_n,\theta_n)$, out of which $L_{\mathrm{estimate}}=10000$ are randomly chosen for evaluation of the corresponding joint CDF (conditional on $\theta_n$). By taking $M_{\mathrm{sampling}}$ larger than $L_{\mathrm{estimate}}$, we are able to estimate all probabilities at higher resolution without needlessly multiplying the number of estimated probabilities for calculating the empirical CDF.

\end{document}